\documentclass[journal,twoside,web]{ieeecolor}
\usepackage[dvipsnames]{xcolor}

\usepackage{generic}
\usepackage{cite}
\usepackage{amsmath,amssymb,amsfonts}
\usepackage{algorithmic}
\usepackage{graphicx}
\usepackage[OT2,T1]{fontenc}
\usepackage{MnSymbol}
\usepackage{lipsum}
\usepackage{color}
\newtheorem{theorem}{Theorem}
\newtheorem{assumption}{Assumption}
\newtheorem{definition}{Definition}

\newtheorem{remark}{Remark}
\DeclareSymbolFont{cyrletters}{OT2}{wncyr}{m}{n}
\DeclareMathSymbol{\Sha}{\mathalpha}{cyrletters}{"58}

\usepackage{textcomp}
\def\BibTeX{{\rm B\kern-.05em{\sc i\kern-.025em b}\kern-.08em
 T\kern-.1667em\lower.7ex\hbox{E}\kern-.125emX}}
\begin{document}

\title{On solving infinite-dimensional\\ Toeplitz Block LMIs}
\author{Flora Vernerey, Pierre Riedinger and Jamal Daafouz\\
\thanks{This work is supported by HANDY project ANR-18-CE40-0010-02.}
\thanks{The authors are with Universit\'e de Lorraine, CNRS, CRAN, F-54000 Nancy, France.}}
\maketitle
\begin{abstract}
This paper focuses on the resolution of infinite-dimensional Toeplitz Block LMIs, which are frequently encountered in the context of stability analysis and control design problems formulated in the harmonic framework. We propose \textcolor{red}{a consistent truncation method that makes this infinite dimensional problem tractable} and demonstrate that a solution to the truncated problem can always be found at any order, provided that the original infinite-dimensional Toeplitz Block LMI problem is feasible. Using this approach, we illustrate how the infinite dimensional solution of a Toeplitz Block LMI based convex optimization problem can be recovered up to \textcolor{red}{an arbitrarily} small error, by solving a finite dimensional truncated problem. The obtained results are applied to stability analysis and harmonic LQR for linear time periodic (LTP) systems. 
\end{abstract}

\section{Introduction}
LMIs are a powerful and versatile tool that can be used to solve a broad range of problems in science and engineering, including control theory, optimization, signal processing, and robotics. One specific type of LMIs is the Infinite-dimensional Toeplitz Block LMIs (TBLMI), which involve matrices of infinite dimension with a Toeplitz block structure. TBLMIs are encountered in the context of harmonic analysis and control, a topic of great theoretical and practical interest in numerous application domains, including energy management and embedded systems to mention few \cite{Farkas,Bolzern,Sanders,Wereley_1990,Zhou,Zhou2008,Almer2}. 

Solving TBLMIs poses a significant challenge due to the infinite dimensionality. The issue we tackle in this paper is different from the problem previously examined in \cite{Ikeda01}, which aimed to reduce an infinite number of LMIs to a finite number of LMIs. In our case, the number of inequalities is finite but the entries and the unknowns are infinite-dimensional. To illustrate the challenges involved, recall the following fact (see \cite{Pierre2022} for more detail): "a truncated matrix of a Hurwitz infinite-dimensional harmonic matrix may not be Hurwitz at any truncation order". As a result, it is possible that solving the truncated version of an infinite-dimensional harmonic Lyapunov equation may not yield a positive definite solution.

In \cite{Pierre2022}, efficient algorithms and methods that leverage the Toeplitz structure have been proposed to determine the infinite-dimensional solution to harmonic Lyapunov or Riccati equations with arbitrarily small error. In this paper, we aim to expand upon this new approach and extend it to the TBLMI framework. \textcolor{red}{To the authors knowledge, it is the first time that this problem is raised.} Our objective is to define a truncated version of the original problem which enables the recovery of the infinite-dimensional solution with arbitrary accuracy.  \textcolor{red}{Contrarily to the literature on the subject \cite{Wereley_1990, Zhou,Zhou2008}, these new results do not invoke Floquet theory. The latter is of interest for stability analysis of LTP systems but it is limited for control design purposes \cite{Pierre2022}}. 

The paper is organized as follows. The next section is dedicated to mathematical preliminaries. In Section III, we define what we call a TBLMI and we give the problem formulation. The main results are established in Section IV where we investigate the truncation of infinite-dimensional TBLMIs that preserves solution positiveness. We show how to recover the infinite-dimensional solution to a TBLMI-based convex optimization problem, to an arbitrarily small error, by solving a finite dimensional truncated problem. We illustrate the results of this paper in section V and apply the proposed procedure to design a harmonic LQR for linear time periodic systems.

{\bf Notations: } The transpose of a matrix $A$ is denoted $A'$ and $A^*$ denotes the complex conjugate transpose $A^*=\bar A'$. The $n$-dimensional identity matrix is denoted $Id_n$. The infinite identity matrix is denoted $\mathcal{I}$. For $m\in\mathbb{Z}^+\cup \{\infty\}$, the flip matrix $J_m$ is the $(2m+1) \times (2m+1)$ matrix having 1 on the anti-diagonal and zeros elsewhere. 
$C^a$ denotes the space of absolutely continuous function,
$L^{p}([a\ b],\mathbb{C}^n)$ (resp. $\ell^p(\mathbb{C}^n)$) denotes the Lebesgues spaces of $p-$integrable functions on $[a, b]$ with values in $\mathbb{C}^n$ (resp. $p-$summable sequences of $\mathbb{C}^n$) for $1\leq p\leq\infty$. $L_{loc}^{p}$ is the set of locally $p-$integrable functions. The notation $f(t)=g(t)\ a.e.$ means almost everywhere in $t$ or for almost every $t$. 
To simplify the notations, $L^p([a,b])$ or $L^p$ will be often used instead of $L^p([a,b],\mathbb{C}^n)$. 
\vspace{-.3cm}
\section{Preliminaries}
\subsection{Infinite dimensional Toeplitz block (\textcolor{red}{TB}) matrices}
 \color{red}
The Toeplitz transformation of a $T-$periodic function $a\in L^{2}([0 \ T], \mathbb{R})$, denoted $\mathcal{T}(a)$, defines a constant Toeplitz and infinite dimensional matrix as follows: 
\begin{align*}
	\mathcal{T}(a)=
	\left[
	\begin{array}{ccccc}
		\ddots & & \vdots & &\udots \\ & a_{0} & a_{-1} & a_{-2} & \\
		\cdots & a_{1} & a_{0} & a_{-1} & \cdots \\
		& a_{2} & a_{1} & a_{0} & \\
		\udots & & \vdots & & \ddots\end{array}\right],\end{align*}
where $(a_{k})_{k\in\mathbb{Z}}$ is the Fourier coefficient sequence of $a$.\\
From the subsequence $a^+=(a_k)_{k>0}$ and $a^-=(a_k)_{k<0}$ of $(a_k)_{k\in\mathbb{Z}}$, we also define the semi-infinite Hankel matrices:
\begin{align*}
	\mathcal{H}(a^+) &= (a_{i+j-1})_{{i,j}>0},\quad \mathcal{H}(a^-) = (a_{-i-j+1})_{{i,j}>0}.\end{align*}
Given an integer $m >0$, we define the $m-$truncation of $\mathcal{T}(a)$, denoted by $\mathcal{T}_m(a)$, the $(2m+1) \times (2m+1)$ principal submatrix of $\mathcal{T}(a)$. We denote  by $\mathcal{H}_{(p,q)}(a^+)$(resp. $\mathcal{H}_{(p,q)}(a^-)$) for any $p,q>0$, the $(2p+1)\times(2q+1)$ Hankel matrix obtained by selecting the first $(2p+1)$ rows and $(2q+1)$ columns of $\mathcal{H}(a^+)$(resp. $\mathcal{H}(a^-)$).
For clarity purpose, we provide in Fig.~\ref{fig20} a block decomposition of an infinite Toeplitz matrix $\mathcal{T}(a)$ to illustrate how the matrices defined above appear. Notice that the flip matrix $J_m$ is defined in the notation part.  
\begin{figure}[h]\begin{center}
		\includegraphics[width=\linewidth]{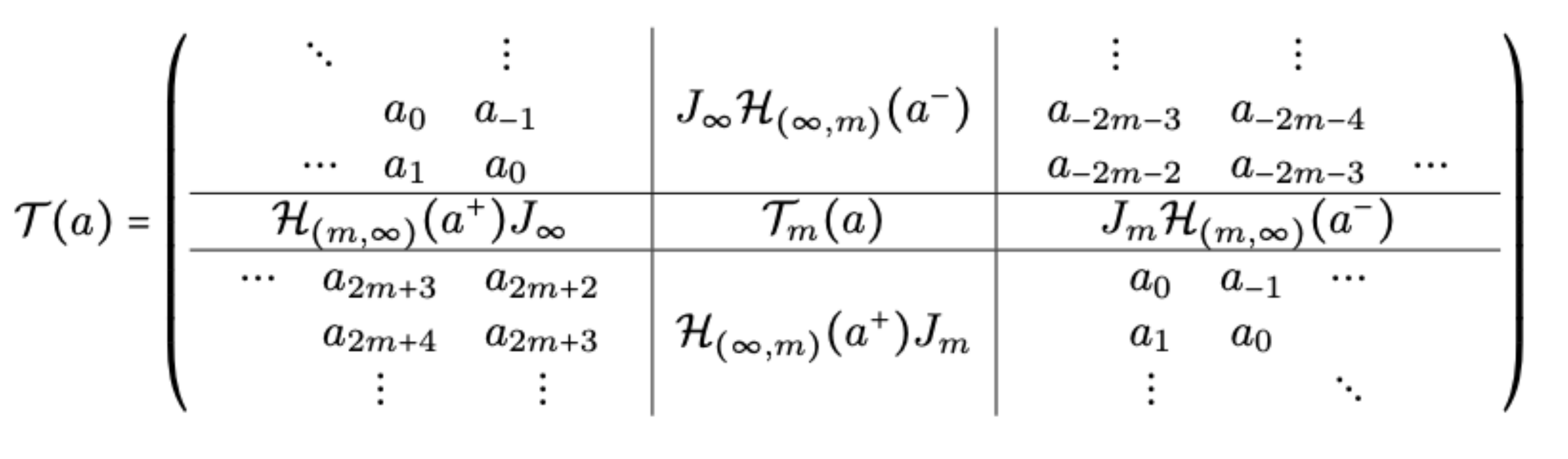}
		\caption{Block decomposition of an infinite Toeplitz matrix $\mathcal{T}(a)$.} \label{fig20}
	\end{center}
\end{figure}
The Toeplitz block transformation of a $T-$periodic $n\times n$ matrix function $A=(a_{ij})_{i,j=1,\cdots,n}\in L^{2}([0 \ T], \mathbb{R}^{n\times n})$, denoted $\mathcal{A}=\mathcal{T}(A)$, defines a constant $n\times n$ Toeplitz Block (\textcolor{red}{TB}) and infinite dimensional matrix: 
\begin{equation}
	\mathcal{A}=\left(\begin{array}{ccc}
		\mathcal{A}_{11} & \cdots & \mathcal{A}_{1n} \\
		\vdots & \ddots& \vdots \\
		\mathcal{A}_{n1} & \cdots & \mathcal{A}_{nn}\end{array}\right)\label{btop}\end{equation} where $\mathcal{A}_{ij}=\mathcal{T}(a_{ij})$, $i,j=1,\cdots,n$.
		
The $m-$truncation of the $n \times n$ Toeplitz block matrix $\mathcal{A}$ denoted by $\mathcal{T}_m(A)$, is defined by the $m-$truncation $\mathcal{T}_m(a_{ij})$ for all its entries $(i,j)$.
Similarly, the $n\times n$ Hankel block matrices $\mathcal{H}(A^+)$, $\mathcal{H}(A^-)$ are also defined respectively by $\mathcal{H}(A^+)_{ij}=\mathcal{H}(a^+_{ij})$ and $\mathcal{H}(A^-)_{ij}=\mathcal{H}(a^-_{ij})$ for $i,j=1,\cdots,n$. {Their principal submatrices $\mathcal{H}(A^+)_{(p,q)}$, $\mathcal{H}(A^-)_{(p,q)}$ for $p>0$, $q>0$ are obtained by considering the principal submatrices of the entries $\mathcal{H}(a^+_{ij})_{(p,q)}$ and $\mathcal{H}(a^-_{ij})_{(p,q)}$ for $i,j=1,\cdots,n$.}\\
\textcolor{red}{Recall that the product of  two \textcolor{red}{TB} matrices is a \textcolor{red}{TB} matrix only in infinite dimension. In finite dimension, we have the following result \cite{Pierre2022}}.
\begin{theorem} \label{product}Let $\mathcal{A}$, $\mathcal{B}$ be two $n \times n$ \textcolor{red}{TB} matrices and $\mathcal{C}=\mathcal{A}\mathcal{B}$. Then,
	\begin{align}
		\mathcal{T}_m(A)\mathcal{T}_m(B) &= \mathcal{T}_m(C)- \mathcal{H}_{(m,\eta)}(A^+)\mathcal{H}_{(\eta,m)}(B^-)\nonumber \\&- \mathcal{J}_{n,m}\mathcal{H}_{(m,\eta)}(A^-)\mathcal{H}_{(\eta,m)}(B^+) \mathcal{J}_{n,m}, \label{ee2}\end{align}
	where $\mathcal{J}_{n,m}=Id_n\otimes J_m$ and $\eta\in \mathbb{Z}^+\cup \{+\infty\}$ is such that $2\eta\geq \min(d^oA,d^oB)$ \textcolor{red}{with $d^o A$ the largest harmonic (non vanishing Fourier coefficient) of $A$}.
\end{theorem}
\begin{figure}\begin{center}
		\includegraphics[scale=0.2]{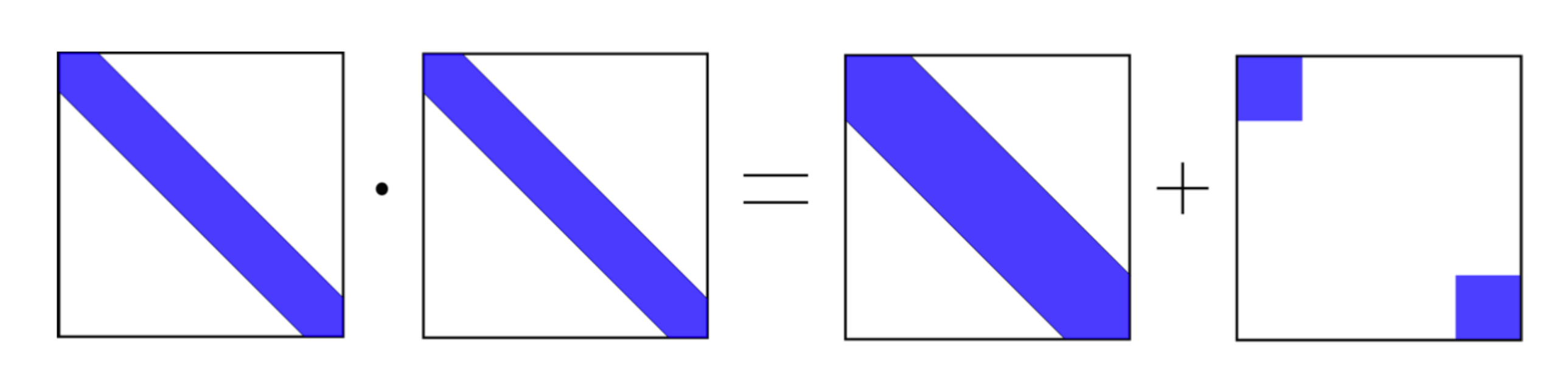}
		\caption{Multiplication of two finite dimensional banded Toeplitz matrices}\label{fig1}
	\end{center}
\end{figure}
An illustration of the above theorem is given in Fig.~\ref{fig1} for $n=1$ when $d^oa$ and
$d^ob$ are less than $m$ so that $\mathcal{T}_m(a)$ and $\mathcal{T}_m(b)$
are banded.
In this case, the matrices $E^+=\mathcal{H}_{(m,m)}(a^+)\mathcal{H}_{(m,m)}(b^-)$ and $E^-= J_m\mathcal{H}_{(m,m)}(a^-)\mathcal{H}_{(m,m)}(b^+) J_m$ have disjoint supports located in the upper leftmost corner and in the lower rightmost corner, respectively. As a consequence,
$\mathcal{T}_m(a)\mathcal{T}_m(b)$ can be represented as the sum of $\mathcal{T}_m(c)$ and two correcting terms $E^+$ and $E^-$.
 \color{black}
\subsection{Sliding Fourier decomposition and harmonic modeling}
Consider $x\in L^{2}_{loc}(\mathbb{R},\mathbb{C})$ a complex valued function of time. Its sliding Fourier decomposition over a window of length $T$ is defined by the time-varying infinite sequence $X=\mathcal{F}(x)\in C^a(\mathbb{R},\ell^2(\mathbb{C}))$ (see \cite{Blin}) whose components satisfy:
$$X_{k}(t)=\frac{1}{T}\int_{t-T}^t x(\tau)e^{-\textsf{j}\omega k \tau}d\tau$$ for $k\in \mathbb{Z}$, with $\omega=\frac{2\pi}{T}$.
If $x=(x_1,\cdots,x_n)\in L^{2}_{loc}(\mathbb{R},\mathbb{C}^n)$ is a complex valued vector function, then
$$X=\mathcal{F}(x)=(\mathcal{F}(x_1), \cdots,\mathcal{F}(x_n)).$$
The vector $X_k=(X_{1,k}, \cdots, X_{n,k})$ with $$X_{i,k}(t)=\frac{1}{T}\int_{t-T}^t x_i(\tau)e^{-\textsf{j}\omega k \tau}d\tau$$
is called the $k-$th phasor of $X$. 
\begin{definition}\label{H} We say that $X$ belongs to $H$ if $X$ is an absolutely continuous function (i.e $X\in C^a(\mathbb{R},\ell^2(\mathbb{C}^n))$ and fulfils for any $k$ the following condition: \begin{equation*}\dot X_k(t)=\dot X_0(t)e^{- \textsf{j}\omega k t} \ a.e.\end{equation*}
\end{definition}
Similarly to the Riesz-Fisher theorem which establishes a one-to-one correspondence between the spaces $L^2$ and $\ell^2$, the following theorem establishes a one-to-one correspondence between the spaces $L_{loc}^2$ and $H$ (see \cite{Blin}).
\begin{theorem}\label{coincidence}For a given $X\in L_{loc}^{\infty}(\mathbb{R},\ell^2(\mathbb{C}^n))$, there exists a representative $x\in L^2_{loc}(\mathbb{R},\mathbb{C}^n)$ of $X$, i.e. $X=\mathcal{F}(x)$, if and only if $X \in H$.
\end{theorem}
\textcolor{red}{Thanks to Theorem~\ref{coincidence}, it is established in \cite{Blin} that any system having solutions in Carath\'eodory sense can be transformed by a sliding Fourier decomposition into an infinite dimensional system for which a one-to-one correspondence between their respective trajectories is established providing that the trajectories in the infinite dimensional space belong to the subspace $H$. Moreover, when a $T-$periodic system is considered, the resulting infinite dimensional system is time-invariant.} For instance, consider $T-$periodic functions $A(\cdot)$ and $B(\cdot)$ respectively of class $L^2([0\ T],\mathbb{C}^{n\times n})$ and $L^{\infty}([0\ T],\mathbb{C}^{n\times m})$ and let: 
\begin{align}\dot x(t)=A(t)x(t)+B(t)u(t)\quad x(0)=x_0\label{ltp}\end{align}
If, $x$ is a solution associated to the control $u\in L_{loc}^2(\mathbb{R},{\mathbb{C}^m)}$ of the linear time periodic (LTP) system ~(\ref{ltp}) then, $X=\mathcal{F}(x)$ is a solution associated to $U=\mathcal{F}(u)$ of the linear time invariant (LTI) system:
\begin{align}
	\dot X(t)=(\mathcal{A}-\mathcal{N})X(t)+\mathcal{B}U(t), \quad X(0)=\mathcal{F}(x)(0) \label{ltih}
\end{align}
where $\mathcal{A}=\mathcal{T}(A)$, $\mathcal{B}=\mathcal{T}(B)$ and 
\begin{equation}\mathcal{N}=Id_n\otimes diag( \textsf{j}\omega k,\ k\in \mathbb{Z})\label{q}\end{equation}
Reciprocally, if $X\in H$ is a solution to \eqref{ltih} with $U\in H$, then their representatives $x$ and $u$
(i.e. $X=\mathcal{F}(x)$ and $U=\mathcal{F}(u)$) are a solution to~\eqref{ltp}. In addition, it is proved in \cite{Blin} that one can reconstruct time trajectories from harmonic ones, that is:
\begin{align*}\label{recos} x(t)&=\mathcal{F}^{-1}(X)(t)=\sum_{k=-\infty}^{+\infty} X_k(t)e^{ \textsf{j}\omega k t}+\frac{T}{2}\dot X_0(t)\end{align*}
where $X_{k}=(X_{1,k}, \cdots, X_{n,k})$ for any $k\in \mathbb{Z}$.
 \textcolor{red}{\remark In this paper, we use a TB matrix representation instead of a more standard Block Toeplitz (BT) matrix representation. The main reason is that it allows to obtain a  structure of the harmonic equations similar to the one in the time domain (see for example \eqref{btop}). This is more suitable for analysis and control design purposes. To obtain a BT structure as in \cite{Blin,Zhou2008}, one has to define $\mathcal{F}$ 
by $X:=\mathcal{F}(x)=( \cdots, X_{-1},X_{0}, X_{1},\cdots)$ where $X_k$ refers to the $k-th$ phasors instead of $X:=\mathcal{F}(x)=(\mathcal{F}(x_1), \cdots,\mathcal{F}(x_n)).$ Obvioulsly, we can always switch from one representation to another by applying a permutation matrix.}
\subsection{Trace operator}
Consider the vectorial space $S^n$ of $T-$periodic, $L^\infty([0\ T])$ and symmetric matrix functions
\textcolor{red}{ and define the scalar product:} 
\begin{align*}
	<M,N>_{S^n}&=\frac{1}{T}\int_0^Ttr(M'(\tau)N(\tau))d\tau\\
&\textcolor{red}{=\frac{1}{T}\sum_{i,j=1}^n\int_0^TM_{ij}(\tau)N_{ij}(\tau))d\tau,}
\end{align*}
\color{red}
for which it is straightforward to show that the induced norm $\|M\|_{S^n}=<M,M>_{S^n}^\frac{1}{2}$ satisfies:\begin{equation}\|M\|_{S^n}\leq n\|M\|_{L^\infty}.\label{nl1}\end{equation}
We recall (see p.p. 562-574 of \cite{Gohberg} for a detailed proof) that $M\in L^\infty([0\ T])$ if and only if $\mathcal{M=}\mathcal{T}(M)$ is bounded on $\ell^2$  i.e. there exists $\kappa>0$, $$\|\mathcal{M}\|_{\ell^2}=\sup_{\|x\|_{\ell^2}=1}\|\mathcal{M}x\|_{\ell^2}<\kappa,$$  
and that the following equality occurs:
\begin{equation}\|\mathcal{M}\|_{\ell^2}=\|M\|_{L^\infty}.\label{nl2}\end{equation}
Let $S^n_+=\{M\in S^n: M\geq 0\ a.e.\}$ and $S^n_{++}=\{M\in S^n: M> 0\ a.e.\}$. 
It follows that: $M\in S^n_{++}$ if and only if $\mathcal{M}=\mathcal{T}(M)$ is Hermitian, positive definite, TB and bounded on $\ell^2$. 
If $M\in S^n$, as any component of $M$ can be rewritten  using its Fourier series (since $L^\infty([0\ T])\subset L^2([0\ T])$:
$$M_{ij}(t)=\sum_{k\in \mathbb{Z}} m_{ij,k} e^{ \textsf{j}\omega kt}\ a.e.,$$
we have: $<Id,M>_{S^n}=\sum_{i=1}^n m_{ii,0}.$
This allows to define the trace operator for $\mathcal{M=}\mathcal{T}(M)$ as follows.
\begin{definition}\label{trace}The trace operator for bounded operators on $\ell^2$ is defined by
	\begin{equation}
		tr(\mathcal{M})=\sum_{i=1}^n m_{ii,0}.\label{tr}\end{equation}
\end{definition}
Then, $tr(\mathcal{M}^*\mathcal{M})^\frac{1}{2}=\|M\|_{S^n}$ defines an operator-norm that satisfies: $$ tr(\mathcal{M}^*\mathcal{M})^\frac{1}{2}\leq n \|\mathcal{M}\|_{\ell^2}.$$
\color{black}
\vspace{-.2cm}
\section{\textcolor{red}{Motivations and }problem formulation}
Before stating the problem we are interested in, we give examples of problems where TBLMIs may be encountered. 
First, consider the problem of stability analysis \textcolor{red}{of the LTP system \eqref{ltp}. In the time domain, this reduces to check the feasibility of the following {\it differential} Lyapunov inequality:\begin{equation}\dot P+A'P+PA<0\label{tlyap}
\end{equation} with $P=P'>0\ a.e.$ and $T-$periodic whereas in the harmonic domain, the problem amounts to checking the feasibility of the following harmonic Lyapunov inequality:
\begin{equation}(\mathcal{A}-\mathcal{N})^*\mathcal{P}+\mathcal{P}(\mathcal{A}-\mathcal{N})<0\label{lyap}
\end{equation}
with $\mathcal{P}=\mathcal{P}^*>0$.}
\color{black}
TBLMIs can also be encountered in control design problems. \textcolor{red}{Consider the state feedback design problem for \eqref{ltp} which consists in the determination of a control: $u(t)=-K(t)x(t)$ where $K(\cdot)$ is a $T-$periodic and $L^\infty$ matrix function. This problem can be approached in an equivalent way in the harmonic domain by determining a TB static gain $\mathcal{K}$ bounded on $\ell^2$ such that the control $U=-\mathcal{K}X$
stabilizes the infinite dimensional harmonic system \eqref{ltih}. The time-domain  control is simply obtained from the formula: $$u(t)=-K(t)x(t)=-\mathcal{F}^{-1}(\mathcal{K}X)(t)$$}
The problem reduces to the determination of  a stabilizing state feedback gain $\mathcal{K}=\mathcal{Y}\mathcal{S}^{-1}$
where the \textcolor{red}{TB} matrices $\mathcal{Y}$ and $\mathcal{S}$ are solutions \textcolor{red}{(bounded on $\ell^2$)} of the TBLMI: 
\begin{align*}
(\mathcal{A}-\mathcal{N})\mathcal{S}+\mathcal{S}(\mathcal{A}-\mathcal{N})^*-\mathcal{B}\mathcal{Y}-\mathcal{Y}^*\mathcal{B}^*&<0 
\end{align*}
with $\mathcal{S}=\mathcal{S}^*>0$.
One may also mention the harmonic LQR problem whose solution is obtained by solving the associated 
infinite dimensional convex optimization problem \cite{Wil71}:
\begin{align} 
 &\max_{\scriptsize \mathcal{P}=\mathcal{P}^*>0} tr(\mathcal{P}),\ 
\label{op}\\
&\left(\begin{array}{cc}
(\mathcal{A}-\mathcal{N})^*\mathcal{P}+\mathcal{P}(\mathcal{A}-\mathcal{N})+\mathcal{Q} & \mathcal{PB} \\
\mathcal{B}^*\mathcal{P}& \mathcal{R}
\end{array}\right)\geq 0\nonumber
\end{align}
where the trace operator is defined by \eqref{tr} and $\mathcal{Q}$ and $\mathcal{R}$ are the LQR weighting matrices. The matrix gain is given by $\mathcal{K}=\mathcal{R}^{-1}\mathcal{B}^*\mathcal{P}$ where $\mathcal{P}$ is a \textcolor{red}{TB} matrix of infinite dimension and a bounded operator on $\ell^2$; see \cite{Blin} for more details. Now, we define what we call a TBLMI.
\begin{definition} A TBLMI is defined by:
\begin{equation}
	\mathcal{L}(\mathcal{P};\textcolor{red}{\mathcal{A}_s,s\in \mathbb{S}})<0
	\label{LMg}
\end{equation}
\textcolor{red}{where $\mathcal{P}$ is the unknown TB operator, $\mathcal{A}_s,s\in \mathbb{S}$ are given TB operators and $\mathbb{S}$ is a finite set of subscripts. Both $\mathcal{P}$ and  $\mathcal{L}(\mathcal{P};\mathcal{A}_s,s\in \mathbb{S})$ are assumed  to be bounded operator on $\ell^2$.\\ For instance, in \eqref{lyap}, we have two given operators $\mathcal{A}_1 = \mathcal{A}$ and $\mathcal{A}_2 = \mathcal{N}$.} 
\end{definition}
\begin{assumption}\label{bound}
All the entries \textcolor{red}{$\mathcal{A}_s,$ $s\in \mathbb{S}$} are bounded operators on $\ell^2$ \textcolor{red}{(equivalently $A_s\in L^{\infty}$ where  $\mathcal{A}_s=\mathcal{T}(A_s)$)}, except $\mathcal{N}$ given by (\ref{q}) which is not. $\mathcal{N}$ appears with the following \textcolor{red}{TB} form: $\mathcal{N}^*\mathcal{P}+\mathcal{P}\mathcal{N}$ \textcolor{red}{which corresponds to the Toeplitz transformation $\mathcal{T}(\dot P)=-\mathcal{N}^*\mathcal{P}-\mathcal{PN}$ where $\mathcal{P}=\mathcal{T}(P)$ and where $P$ is a $T-$periodic and absolutely continuous matrix function \cite{Blin}.}
\end{assumption}
\color{black}
\textcolor{red}{The problem we tackle in this paper is to determine a solution of the following Convex Optimisation Problem ({\bf COP})}:
\begin{align*}{\bf COP:} &\min_{\mathcal{P}^*=\mathcal{P}>0} tr(\mathcal{P})\text{ subject to:}\\
	&\mathcal{L}(\mathcal{P};\textcolor{red}{\mathcal{A}_{s}, s\in \mathbb{S}})\leq 0
\end{align*}
\textcolor{red}{
We assume that this convex optimization problem is feasible and that the optimal solution is unique, bounded on $\ell^2$ and continuous with respect to the entries $\mathcal{A}_s, s\in\mathbb{S}$.
} 
\vspace{-.05cm}
\section{Main results}
\vspace{-.05cm}
\textcolor{red}{  ${\bf COP}$ is an infinite-dimensional problem in the sense that the dimension of the involved entries and unknowns is infinite. The main objective here is to show how ${\bf COP}$ can be solved up to an arbitrarily small error. The main results are presented in three steps. The first one defines what is a truncated TBLMI. The second step shows how to combine truncation and banded approximation operations in order to obtain a finite-dimensional problem. The third step proves that the solution of ${\bf COP}$ can be recovered up to an arbitrary small error by solving a finite optimization problem. }
\vspace{-.2cm}
\subsection{Truncation operator $\Pi$}
\begin{definition}\label{proj} Consider infinite-dimensional \textcolor{red}{TB} matrices $\mathcal{A}$ and $\mathcal{B}$ of compatible size. The truncation operator $\Pi_m$ at order $m$ is determined by:
\begin{align}
&\Pi_m(\mathcal{A}) =\mathcal{T}_m(A),\nonumber\\
&\Pi_m(\mathcal{A}+\mathcal{B})=\Pi_m(\mathcal{A}) +\Pi_m(\mathcal{B})\nonumber\\
&\Pi_m(\mathcal{AB})= \Pi_m(\mathcal{A})\Pi_m(\mathcal{B})+\mathcal{H}_{(m,\eta)}(A^+)\mathcal{H}_{(\eta,m)}(B^-) \nonumber  \\&\qquad \qquad \quad+\mathcal{J}_{n,m} \mathcal{H}_{(m,\eta)}(A^-)\mathcal{H}_{(\eta,m)}(B^+)\mathcal{J}_{n,m}\label{pro}
\end{align}
where $\eta\in \mathbb{Z}^+\cup \{+\infty\}$ is such that $2\eta\geq \min \textcolor{red}{(d^oA,d^oB)}$.
\end{definition}
For a given $m$, the $m-$truncated TBLMI of \eqref{LMg} is:
 \begin{equation}\Pi_m(\mathcal{L}(\mathcal{P};\textcolor{red}{\mathcal{A}_{s}, s\in \mathbb{S}}))<0\label{lmi_trunc}\end{equation}
For example, the $m-$truncated TBLMI associated to \eqref{lyap}~is: 
\begin{align}&\textcolor{red}{\Pi_m((\mathcal{A}-\mathcal{N})^*)\Pi_m(\mathcal{P})+\Pi_m(\mathcal{P})\Pi_m(\mathcal{A}-\mathcal{N})}\nonumber\\
& +\mathcal{H}_{(m,\eta)}(A^{*+})\mathcal{H}_{(\eta,m)}(P^-)+\mathcal{H}_{(m,\eta)}(P^{+})\mathcal{H}_{(\eta,m)}(A^{-})\nonumber\\
 &+\mathcal{J}_{n,m} (\mathcal{H}_{(m,\eta)}(A^{*-})\mathcal{H}_{(\eta,m)}(P^+)\nonumber\\&+ \mathcal{H}_{(m,\eta)}(P^-)\mathcal{H}_{(\eta,m)}(A^+))\mathcal{J}_{n,m}<0\nonumber
\end{align}
with $\eta\in \mathbb{Z}^+\cup \{+\infty\}$ is such that $2\eta\geq \min  \textcolor{red}{(d^oA,d^oP)}$.
\begin{theorem}\label{sol_trunc}If $\mathcal{P}$ solves the infinite-dimensional TBLMI (\ref{LMg}) then $\mathcal{P}$ 
solves the truncated TBLMI (\ref{lmi_trunc}) at any order~$m$.
\end{theorem}
\begin{proof} Consider a solution $\mathcal{P}$ to \eqref{LMg} then for any $m>0$, the principal submatrix $\Pi_m(\mathcal{L}(\mathcal{P};\textcolor{red}{\mathcal{A}_s,s\in \mathbb{S}}))$ \textcolor{red}{of $\mathcal{L}(\mathcal{P};\textcolor{red}{\mathcal{A}_s,s\in \mathbb{S}})$} is necessarily negative definite. Moreover, \textcolor{red}{thanks} to \eqref{ee2} and to Definition~\ref{proj},  $\Pi_m(\mathcal{L}(\mathcal{P};\textcolor{red}{\mathcal{A}_s,s\in \mathbb{S}}))$ \textcolor{red}{can be explicitly developed without any approximation.} Therefore, a solution to the obtained $m-$truncated TBLMI can be deduced from $\mathcal{P}$ itself.
\end{proof}
\textcolor{red}{Consequently}, if the infinite-dimensional TBLMI (\ref{LMg}) is feasible then there always exists a solution to the truncated TBLMI (\ref{lmi_trunc}) at any order $m$. 
Unfortunately, \eqref{lmi_trunc} may contain terms involving infinite-dimensional Hankel matrices (when $\eta=+\infty$ in \eqref{pro}).  \textcolor{red}{The next section shows how this infinite-dimensional problem can be reduced to a finite one by considering banded approximation of the entries.}
\subsection{Truncated and banded approximation of \textcolor{red}{TBLMI}}
 \textcolor{red}{The aim of this part is to show that \eqref{LMg} can be approximated by a banded version (see \eqref{blmi}) whose $m$-truncation (see \eqref{lmi_trunc3}) is now tractable numerically since only a finite number of unknowns must be taken into account.}
\color{red} To this end, we define in the sequel, for any \textcolor{red}{TB} operator $\mathcal{A}$, 
 its $p-$banded version denoted by $\mathcal{A}_{b(p)}$ and obtained by deleting all its phasors of order higher than $p$.
%
\begin{theorem}\label{tt1}The following results hold true:
\begin{enumerate}
\item Assume that $\mathcal{A}$ is a bounded operator on $\ell^2$. The operator $\mathcal{A}_{b(p)}$ converges to $\mathcal{A}$ in $\ell^2$-operator norm i.e.
$$\lim_{p\rightarrow +\infty}\|\mathcal{A}-\mathcal{A}_{b(p)}\|_{\ell^2}=0$$

\item Under Assumption \ref{bound}, if ${\mathcal{P}}$ is a solution to \eqref{LMg} then there exists $p_0$ such that for any $p\geq p_0$, 
\begin{equation}\mathcal{L}({\mathcal{P}};\textcolor{red}{\mathcal{A}_{s_{b(p)}}, s\in \mathbb{S}})<0 \label{blmi}\end{equation} 
\item For given $p$ and $m$, the $m-$truncated and $p-$banded TBLMI: 
\begin{equation} 
\Pi_m(\mathcal{L}(\mathcal{P};\textcolor{red}{\mathcal{A}_{s_{b(p)}}, s\in \mathbb{S}}))<0 \label{lmi_trunc3}.
\end{equation}
 involves  a finite number of unknowns. 
 \end{enumerate}
\end{theorem}
\color{black}
\begin{proof}
\textcolor{red}{Let us show the first assertion. As $\|\mathcal{A}\|_{\ell^2}=\|A\|_{L^\infty}$ where $\mathcal{A}=\mathcal{T}(A)$ (see \eqref{nl2}) and using the Fourier series of $A$: $$A(t)=\sum_{k\in\mathbb{Z}} A_ke^{ \textsf{j}\omega kt} \ a.e.,$$} we can write: 
\begin{align}
\|\mathcal{A}-\mathcal{A}_{b(p)}\|_{\ell^2}& =\|A-A_{b(p)}\|_{L^\infty}=\|\sum_{|k|>p} A_ke^{ \textsf{j}\omega kt} \|_{L^\infty}\label{e1}
\end{align}
As by assumption there exists a constant $C_1$ such that
\begin{align*}
\|\mathcal{A}\|_{\ell^2}=\|A\|_{L^\infty}
&= \|\sum_{k\in \mathbb{Z}} A_ke^{ \textsf{j}\omega kt} \|_{L^\infty}<C_1
\end{align*}
the series $\sum_{k\in \mathbb{Z}} A_ke^{ \textsf{j}\omega kt}$ converges almost everywhere and $\lim_{p\rightarrow +\infty}\sum_{|k|>p} A_ke^{ \textsf{j}\omega kt}=0\ a.e.$ 
Taking the limit w.r.t. $p$ in \eqref{e1} leads to the result. \\
\textcolor{red}{Now for assertion 2),} by Assumption~\ref{bound}, the only entry of the TBLMI not bounded on $\ell^2$ is $\mathcal{N}$. Fortunately, as $\mathcal{N}$ is diagonal, $\mathcal{N}=\mathcal{N}_{{b(p)}}$ for any $p\geq 0$ and thus $\mathcal{N}$ does not play any role.
If ${\mathcal{P}}$ is a solution to \eqref{LMg}, then by assumption $\mathcal{L}({\mathcal{P}};\textcolor{red}{\mathcal{A}_{s}, s\in \mathbb{S}})$ must be a bounded operator on $\ell^2$. LMIs being continuous with respect to their entries, there exists a constant $C_2$ depending of ${\mathcal{P}}$ and \textcolor{red}{$\mathcal{A}_s,$ $s\in \mathbb{S}$} such that for any $p>0$
\begin{align*}\|\mathcal{L}({\mathcal{P}};\textcolor{red}{\mathcal{A}_{s}, s\in \mathbb{S}})-\mathcal{L}({\mathcal{P}};&\textcolor{red}{\mathcal{A}_{s_{b(p)}}, s\in \mathbb{S}})\|_{\ell^2}\\&\leq C_2 \sum_{s\in \textcolor{red}{\mathbb{S}}} \|\mathcal{A}_{s}-\mathcal{A}_{s_{b(p)}}\|_{\ell^2}\end{align*}
From the first assertion, we conclude that for any $\epsilon>0$, there exists $p_0$ such that for $p\geq p_0$, 
$$\|\mathcal{L}({\mathcal{P}};\textcolor{red}{\mathcal{A}_{s}, s\in \mathbb{S}})-\mathcal{L}({\mathcal{P}};\textcolor{red}{\mathcal{A}_{s_{b(p)}}, s\in \mathbb{S}})\|_{\ell^2}<\epsilon$$
and relation \eqref{blmi} follows for sufficiently small $\epsilon$.\\
Finally, to show the last assertion, as all \textcolor{red}{$\mathcal{A}_{s_{b(p)}}$, $s\in \mathbb{S}$} in $\mathcal{L}(\mathcal{P};\textcolor{red}{\mathcal{A}_{s_{b(p)}}, s\in \mathbb{S}})<0$ are banded, only the unknown $\mathcal{P}$ is possibly not banded.
As the product of infinite dimensional banded \textcolor{red}{TB} operators is a banded \textcolor{red}{TB} operator (which is not true in finite dimension), the terms in the TBLMI involving operator $\mathcal{P}$ have the generic form: 
$\mathcal{U}\mathcal{P}\mathcal{V}$ where $\mathcal{U}$ and $\mathcal{V}$ are polynomial functions of banded entries $\textcolor{red}{\mathcal{A}_{s_{b(p)}}, s\in \mathbb{S}}$ and are therefore banded. 
Applying $\Pi_m$ on $\mathcal{L}(\mathcal{P};\textcolor{red}{\mathcal{A}_{s_{b(p)}}, s\in \mathbb{S}})$ leads to compute $\Pi_m(\mathcal{U}\mathcal{P}\mathcal{V})$. Using \eqref{pro}, we have:
\begin{align}
\Pi_m(\mathcal{U}\mathcal{P}\mathcal{V})&= \Pi_m(\mathcal{U})\Pi_m(\mathcal{PV})\label{ee1} \\
&+\mathcal{H}_{(m,\eta_1)}(U^+)\mathcal{H}_{(\eta_1,m)}((PV)^-) \nonumber \\
&+\mathcal{J}_{n,m} \mathcal{H}_{(m,\eta_1)}(U^-)\mathcal{H}_{(\eta_1,m)}((PV)^+)\mathcal{J}_{n,m}\nonumber 
\end{align}
where $\eta_1$ is the first integer greater than $\frac{1}{2} d^o U$ and where $\Pi_m(\mathcal{PV)}$ is determined using \eqref{pro} with $\eta$ the first integer greater than $\frac{1}{2} d^o V$. Noticing that the coefficient of \textcolor{red}{highest} degree invoked in the Hankel matrix $\mathcal{H}_{(m,\eta)}(\cdot)$ is of degree $2(m+\eta)+1$, it is straightforward to check that only a finite number of phasors of $\mathcal{P}$ are necessary to compute both $\Pi_m(\mathcal{PV)}$ and \eqref{ee1} and thus the result follows. 
\end{proof}

\vspace{-.1cm}
\subsection{Solving {\bf COP} up to an arbitrary error}
\textcolor{red}{
We are now ready to prove the main result of this paper. To this end, we define three subproblems: the $p-$banded problem ${\bf COP_{p}}$, the fully banded problem ${\bf COP_{p,q}} $ and the $m-$truncated, fully banded ${\bf COP_{p,q,m}}$. The main result states that solving ${\bf COP_{p,q,m}}$ is a consistent scheme allowing to approximate the solution to ${\bf COP}$.}\\ 
Consider for a given $p>0$, the $p-$banded problem is: 
\begin{align*}
{\bf COP_{p}:} &\min_{\mathcal{P}^*=\mathcal{P}>0} tr(\mathcal{P}) \text{ subject to:}\\
&\mathcal{L}(\mathcal{P};\textcolor{red}{\mathcal{A}_{s_{b(p)}}, s\in \mathbb{S}})\leq 0.\nonumber
\end{align*}
For a given $q>0$, the fully banded problem is: 
\begin{align*}&{\bf COP_{p,q}:} \min_{\mathcal{P}^*=\mathcal{P}>0} tr(\mathcal{P}) \text{ subject to:} \\
&\mathcal{L}(\mathcal{P};\textcolor{red}{\mathcal{A}_{s_{b(p)}}, s\in \mathbb{S}})\leq 0, \quad
P_{ij,k}=0, |k|>q,\ i,j=1,\cdots,n\nonumber
\end{align*}
\textcolor{red}{where for a given $k\in\mathbb{Z}$, $P_{ij,k}=0$ refers to the $k$th-phasors of the $(i,j)$th block of $\mathcal{P}$.}
For a given $m>0$, the $m-$truncated, fully banded optimization problem is: 
\begin{align*}
&{\bf COP_{p,q,m}:} \min_{\mathcal{P}^*=\mathcal{P}} tr(\mathcal{P}) \text{ subject to:}\quad \Pi_m(\mathcal{P})>0,\\
&\Pi_m(\mathcal{L}(\mathcal{P};\textcolor{red}{\mathcal{A}_{s_{b(p)}}, s\in \mathbb{S}}))\leq 0,\ 
P_{ij,k}=0, |k|>q,\ i,j=1,\cdots,n.\nonumber
\end{align*}
\textcolor{red}{Noting that $tr(\mathcal{P})$ requires only a finite number of phasors to be evaluated, only $ {\bf COP_{p,q,m}}$ is a finite-dimensional problem.}

We assume that all these convex optimization problems are feasible and that the optimal solution is unique, bounded on $\ell^2$ and continuous with respect to the entries $\mathcal{A}_s$, $s\in \mathbb{S}$. 
\textcolor{red}{\begin{assumption}\label{as2}
For given $p,q,m>0$, any unbounded  sequence on $\ell^2$ of admissible candidates for Problem ${\bf COP_{p,q,m}}$ has an unbounded objective function.\end{assumption}
Given $p,q$ and $m$, we denote by $\hat{\mathcal{P}}$, $\hat{\mathcal{P}}_{p}$, $\hat{\mathcal{P}}_{p,q}$ and $\hat{\mathcal{P}}_{p,q,m}$, the solution to ${\bf COP}$, ${\bf COP_{p}}$, ${\bf COP_{p,q}}$ and ${\bf COP_{p,q,m}}$ respectively.
The next theorem states that solving ${\bf COP_{p,q,m}}$ is a consistent scheme allowing to approximate the solution to ${\bf COP}$.
\begin{theorem}\label{prop} For any $\epsilon>0$, there exist $p$, $q$ and $m_0$ such that  for any $m>m_0$:
\begin{align}
& \|\hat{\mathcal{P}}_{p,q,m}-\hat{\mathcal{P}}\|_{\ell^2}< \epsilon.
\end{align}
\end{theorem}}
\begin{proof}
\textcolor{red}{Let us show that the following three limits hold:
\begin{align*}
&\lim_{p \rightarrow +\infty} \|\hat{\mathcal{P}}_{p}-\hat{\mathcal{P}}\|_{\ell^2}=0, \quad
 \lim_{q \rightarrow +\infty} \|\hat{\mathcal{P}}_{p,q}-\hat{\mathcal{P}}_p\|_{\ell^2}=0,\ p>0\\
& \text{and }\lim_{m \rightarrow +\infty} \|\hat{\mathcal{P}}_{p,q,m}-\hat{\mathcal{P}}_{p,q}\|_{\ell^2}=0, p,q>0.
\end{align*}
The first limit} is a direct consequence of the continuity of the optimal solution with respect to the entries \textcolor{red}{$\mathcal{A}_s$, $s\in \mathbb{S}$} and 1) in Theorem~\ref{tt1}.
\color{red}To prove the second limit, for a given $p$, as for any $q$, $\hat{\mathcal{P}}_{p,q}$ is admissible for both ${\bf COP_{p,q+1}}$ and ${\bf COP_{p}}$, it follows necessarily that
\begin{equation}tr(\hat{\mathcal{P}}_{p,q})\geq tr(\hat{\mathcal{P}}_{p,q+1})\geq\cdots \geq tr(\hat{\mathcal{P}}_{p}) >0\label{des}\end{equation}
On the other hand, from 1) in Theorem~\ref{tt1}, it is clear that there exists $q_0$ such that for any $q>q_0$ the $q-$banded operator $\hat{\mathcal{P}}_{p_{b(q)}}$ of $\hat{\mathcal{P}}_{p}$ is positive definite. Thus, $\hat{\mathcal{P}}_{p_{b(q)}}$ is then obviously admissible for  ${\bf COP_{p,q}}$, it follows that: 
\begin{equation}
tr(\hat{\mathcal{P}}_{p_{b(q)}})\geq tr(\hat{\mathcal{P}}_{p,q})\geq tr(\hat{\mathcal{P}}_{p})\label{tg}\end{equation}
Since $\|\hat{\mathcal{P}}_{p_{b(q)}}-\hat{\mathcal{P}}_{p}\|_{\ell^2}\rightarrow0$ when $q\rightarrow+\infty$, taking the limit w.r.t. $q$  in \eqref{tg} leads to:
\begin{equation}\lim_{q\rightarrow+\infty}tr(\hat{\mathcal{P}}_{p,q})= tr(\hat{\mathcal{P}}_{p}).\label{tg2}\end{equation}
 Now, let us show that we also have: $\lim_{q\rightarrow+\infty}\hat{\mathcal{P}}_{p,q}= \hat{\mathcal{P}}_{p}$ on $\ell^2$.
As $ \hat{\mathcal{P}}_{p,q}$ is TB, Hermitian, positive definite and bounded on $\ell^2$, there exists a bounded operator on $\ell^2$, $\mathcal{Z}_{p,q}$ such that the following decomposition holds:$$ \hat{\mathcal{P}}_{p,q}=\mathcal{Z}^*_{p,q}\mathcal{Z}_{p,q} \text{ for any }p,q.$$
Moreover as $\mathcal{Z}_{p,q}$ is a constant matrix function, it belongs trivially in $H$ (see Def.~\ref{H}) and there exists a representative $Z_{p,q}\in L^\infty([0\ T])$ (see \eqref{nl2}) such that
$\mathcal{Z}_{p,q}=\mathcal{T}(Z_{p,q})$.
Using similar arguments, $\hat{\mathcal{P}}_{p}=\mathcal{Z}^*_{p}\mathcal{Z}_{p}$ with $\mathcal{Z}_{p}=\mathcal{T}(Z_{p})$ and $Z_{p}\in L^\infty([0\ T])$. 
Therefore, Def. 2 implies: $$tr( \hat{\mathcal{P}}_{p,q})=tr(\mathcal{Z}^*_{p,q}\mathcal{Z}_{p,q})=<Z_{p,q},Z_{p,q}>_{S^n}$$
and from \eqref{tg2}, it can be concluded that 
\begin{equation}\lim_{q\rightarrow+\infty}<Z_{p,q},Z_{p,q}>_{S^n}=<Z_{p},Z_{p}>_{S^n}.\label{nc}
\end{equation}
Moreover as the sequence $Z_{p,q}$ indexed by $q$ is bounded (see \eqref{des}), there exists a subsequence that converges weakly on $L^\infty$ and 
Eq. \eqref{nc} implies that it also converges strongly and necessarily to $Z_{p}$ by uniqueness of solution. Finally, the uniqueness of the solution implies that the whole sequence converges to  $Z_{p}$. It follows 
that  $\lim_{q\rightarrow+\infty} \hat {\mathcal{P}}_{q,p}=\hat{\mathcal{P}}_{p}=\mathcal{Z}^*_{p}\mathcal{Z}_{p}$ on $\ell^2$.
Finally to prove the third limit: \color{black} for any $m$ and following similar steps of the proof of Theorem~\ref{sol_trunc} as $\hat{\mathcal{P}}_{p,q}$ is admissible for Problem ${\bf COP_{p,q,m}} $ and $\hat{\mathcal{P}}_{p,q,m+1}$ is admissible for Problem ${\bf COP_{p,q,m}}$ , it follows that:\textcolor{red}{
$$tr(\hat{\mathcal{P}}_{p,q,m})\leq tr(\hat{\mathcal{P}}_{p,q,m+1})\leq \cdots \leq tr(\hat{\mathcal{P}}_{p,q})$$}
which proves that the sequence $tr(\hat{\mathcal{P}}_{p,q,m})$ indexed by $m$ is an increasing and bounded real sequence, and thus a converging sequence. Moreover, for any $m_0$, as the sequence $\hat{\mathcal{P}}_{p,q,m}$, $m\geq m_0$ is admissible for Problem ${\bf COP_{m_0,p,q}}$, following \textcolor{red}{Assumption~\ref{as2},} this sequence is necessarily bounded on $\ell^2$.
Therefore, for any $m\geq m_0$, the phasors of $\hat{\mathcal{P}}_{p,q,m}$ are bounded and belong to a finite dimensional subspace of $\ell^2$ (thanks to the constraints $P_{ij,k}=0 \text{ for } |k|>q$). 
By compactness, there exists a subsequence that converges on this finite subspace of $\ell^2$ and the uniqueness of the solution implies that the whole sequence converges necessarily to $\hat{\mathcal{P}}_{p,q}$.\\
\textcolor{red}{The final result follows since $\forall\epsilon>0$, $\exists p_0$, $\forall p>p_0$, $\exists q_0(p)$, $\forall q>q_0$, $\exists m_0(p,q)$ such that:
$ \|\hat{\mathcal{P}}_{p}-\hat{\mathcal{P}}\|_{\ell^2}\leq \frac{\epsilon}{3}$, $\|\hat{\mathcal{P}}_{p,q}-\hat{\mathcal{P}}_p\|_{\ell^2}\leq \frac{\epsilon}{3}$ and $\forall m>m_0$ $\|\hat{\mathcal{P}}_{p,q,m}-\hat{\mathcal{P}}_{p,q}\|_{\ell^2}\leq \frac{\epsilon}{3}$
and thus it follows that 
\begin{align*}
\|\hat{\mathcal{P}}_{p,q,m}-\hat{\mathcal{P}}\|_{\ell^2}\leq&\|\hat{\mathcal{P}}_{p,q,m}-\hat{\mathcal{P}}_{p,q}\|_{\ell^2}\\&+\|\hat{\mathcal{P}}_{p,q}-\hat{\mathcal{P}}_p\|_{\ell^2}+ \|\hat{\mathcal{P}}_{p}-\hat{\mathcal{P}}\|_{\ell^2}\leq \epsilon\end{align*}}
\end{proof}
\color{black}
\vspace{-.2cm}
\section{Illustrative example}
We consider the example given in \cite{Pierre2022} defined by:
\begin{align}
	\dot x=&\left(\begin{array}{cc}a_{11} (t) & a_{12} (t) \\a_{21} (t) & a_{22} (t)\end{array}\right)x+\left(\begin{array}{c}b_{11}(t) \\0\end{array}\right)u\label{ex_ltp}\end{align}
{\small\begin{align*}a_{11} (t) &=1+\frac{4}{\pi}\sum_{k=0}^{\infty}\frac{1}{2k+1}\sin(\omega (2k+1)t),\\
	a_{12} (t) &= 2+\frac{16}{\pi^2}\sum_{k=0}^{\infty}\frac{1}{(2k+1)^2}\cos(\omega (2k+1)t),\\
	a_{21} (t) &= -1+\frac{2}{\pi}\sum_{k=1}^{\infty}\frac{(-1)^k}{k}\sin(\omega kt+\frac{\pi}{4}),\\
	a_{22} (t) &= 1-2\sin(\omega t)-2\sin(3\omega t)+2\cos(3\omega t)+2\cos(5\omega t),\\
	b_{11}(t)&=1+ 2 \cos(2\omega t)+ 4 \sin(3\omega t) \text{ with }\omega=2\pi.
\end{align*}}
The associated Toeplitz matrix $\mathcal{A}$ has an infinite number of phasors and is not banded. This system is unstable and the equivalent harmonic LTI system \eqref{ltih} has a spectrum provided by the set $\sigma=\{\lambda+ \textsf{j}\omega k, k\in \mathbb{Z}\}$ where $\lambda \in \{1\pm \textsf{j} 1.64\}$ (see \cite{Pierre2022}).

We consider the LQ problem and we solve the optimization problem given by \eqref{op} with $\mathcal{Q}=10^2\mathcal{T}(Id_n)$ and $\mathcal{R}=\mathcal{T}(Id_m)$. Imposing \textcolor{red}{as required }a \textcolor{red}{TB} structure to $\mathcal{P}$, Problem ${\bf COP_{p,q,m}}$ associated to \eqref{op} is solved with $m=10,15,20$ and $p=q=2m$ \textcolor{red}{(Obviously other choices are possible such as for exemple $m=p=q$)}. This is illustrated in Fig.~\ref{f1} where we plot the modulus of phasors of $\mathcal{K}=[\mathcal{K}_1,\mathcal{K}_2]$ \textcolor{red}{and in Fig.~\ref{f4} where the control $u(t):=-K(t)x(t)$ with $K(t)$ the $T-$periodic gain matrix given by $K(t):=\sum_{k=-2m}^{2m} K_k e^{\textsf{j}\omega kt}$, stabilizes globally and asymptotically the unstable LTP system \eqref{ex_ltp}. As a result, we recover the same state feedback gain as in \cite{Pierre2022} which was obtained using a Kleinman-like algorithm.}
\begin{figure}[h]
	\begin{center}
		\includegraphics[width=\linewidth,height=4.7cm]{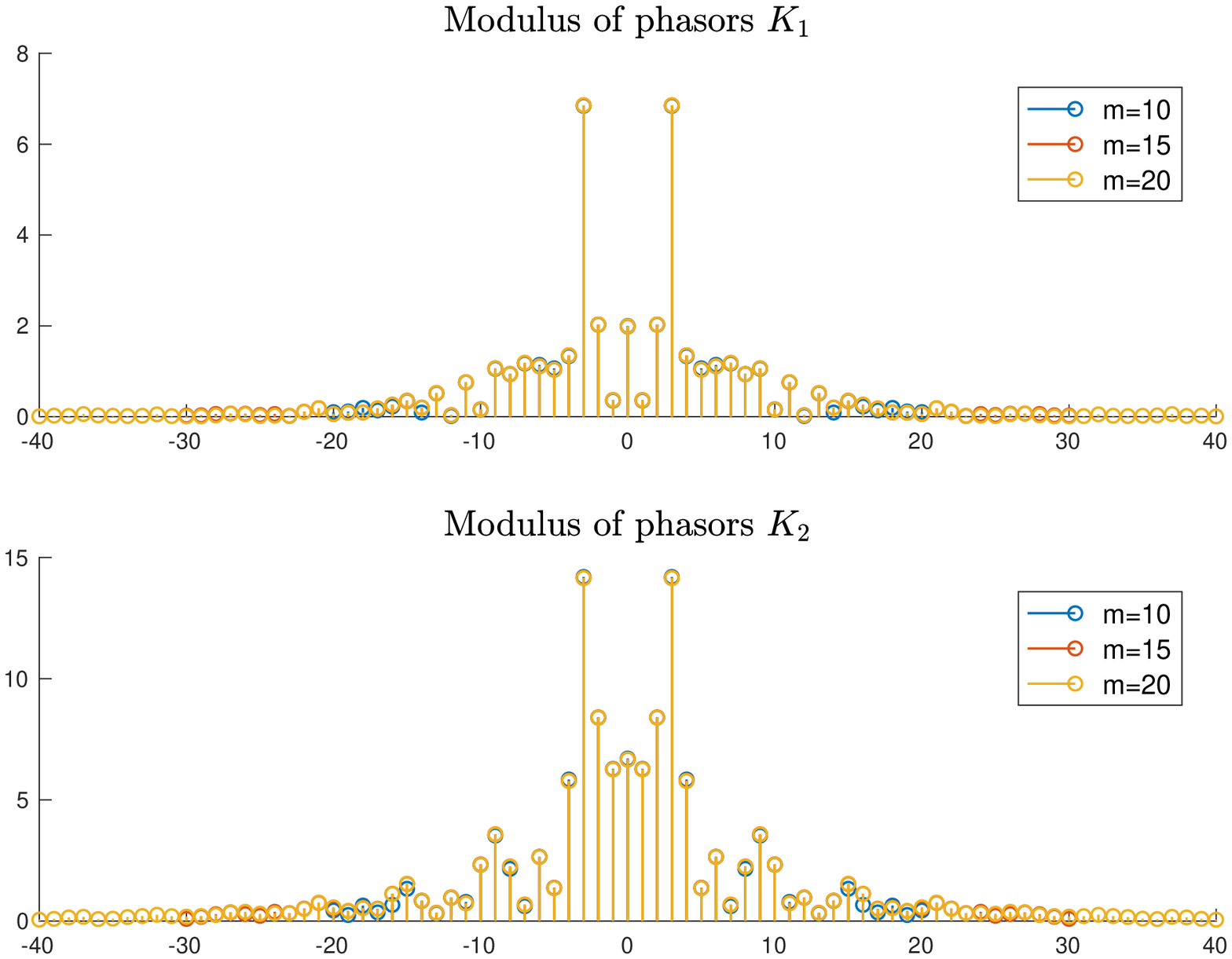}
		\caption{Modulus of Phasors $K=[K_1,K_2]$ (harmonic LQ control)}\label{f1}
	\end{center}
\end{figure}
\begin{figure}[h]
	\begin{center}
		\includegraphics[width=\linewidth,height=5.5cm]{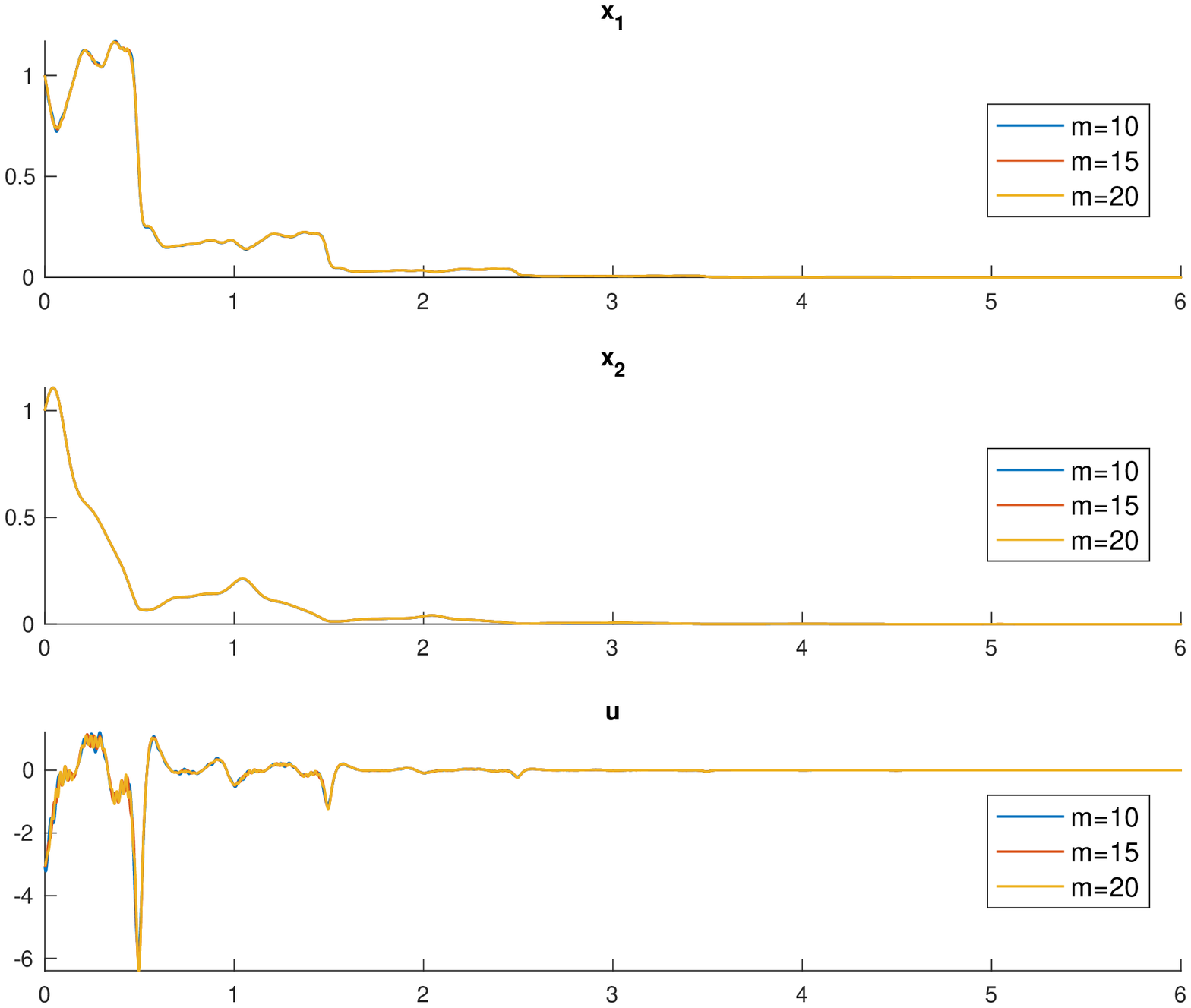}
		\caption{Closed loop response with LQ control}\label{f4}
	\end{center}
\end{figure}
\section{Conclusion}
In this paper, we provided a novel approach that allows to solve, up to an arbitrarily small error, infinite dimensional \textcolor{red}{TB}LMIs and \textcolor{red}{some} related convex optimization problems encountered in the analysis and control of dynamical systems in the harmonic framework. The result is based on a well-defined finite dimensional truncated problem that allows to recover the original infinite-dimensional solution up to an arbitrarily small error. This framework is not only useful for robustness and multiobjective optimization issues of LTP systems but also for the analysis and control of more general periodic systems such as periodic polynomial systems. 

\end{document}